\documentclass[11pt]{article}
\usepackage[english]{babel}
\usepackage{amsmath,amsthm}
\usepackage{amsfonts}
\usepackage[comma,numbers,square,sort&compress]{natbib}
\usepackage{graphicx,subfigure,amsmath,amssymb,mathrsfs,amsfonts,amstext,amsthm}
\usepackage{latexsym, amssymb}
\usepackage{graphicx}
\usepackage{subfigure,wrapfig,caption,lipsum}
\usepackage{pgf,tikz}
\usepackage{pstricks}
\newtheorem{thm}{Theorem}[section]

\newtheorem{lem}[thm]{Lemma}

\theoremstyle{definition}
\newtheorem{defn}[thm]{Definition}
\theoremstyle{remark}
\newtheorem{rem}[thm]{Remark}
\numberwithin{equation}{section}
\begin{document}

\title{\bfseries\textrm{Noise-driven Synchronization of Vicsek Model in Mean}
\footnotetext{*Wei Su and Yongguang Yu are with School of Mathematics and Statistics, Beijing Jiaotong University, Beijing 100044, China, {\tt su.wei@bjtu.edu.cn, ygyu@bjtu.edu.cn}. Ge Chen is with National Center for Mathematics and Interdisciplinary Sciences \& Key Laboratory of Systems and
Control, Academy of Mathematics and Systems Science, Chinese Academy of Sciences, Beijing 100190,
China, {\tt chenge@amss.ac.cn} }
}

\author{Wei Su, Yongguang Yu, Ge Chen}

\date{}
\maketitle
\begin{abstract}
The Vicsek model has long stood as a pivotal framework in exploring collective behavior and self-organization, captivating the scientific community with its compelling dynamics. However, understanding how noise influences synchronization within this model and its associated phase transition characteristics has presented significant challenges. While numerous studies have focused on simulations due to the model's mathematical complexity, comprehensive theoretical analyses remain sparse. In this paper, we deliver a rigorous mathematical proof demonstrating that for any initial configuration of the Vicsek model, there exists a bound on noise amplitude such that if the noise amplitude is maintained within this bound, the system will achieve synchronization in mean. This finding not only lays a solid mathematical groundwork for the Vicsek model's phase transition theory but also underscores the critical role of noise in collective dynamics, enhancing our understanding of self-organizing systems in stochastic environments.
\end{abstract}

\section{Introduction}\label{Section:introduction}
Collective behavior in biological systems has garnered significant attention in recent years, as it encompasses a variety of phenomena observed in nature, including flocking of birds, schooling of fish, and swarming of insects \cite{Haken1983,Toner1995,Couzin2005,Bialek2012,Lowen2023,Yan2024}. To gain deeper insights into these complex systems, the Vicsek model has gradually emerged as a prominent framework. Initially proposed by Vicsek et al. in 1995, this model provides a simplified yet powerful mathematical description of self-organizing systems  \cite{Vicsek1995}. In the Vicsek model, particles (agents) move in a two-dimensional space, aligning their velocities based on the average direction of their neighboring agents. Despite its simple design, the Vicsek model exhibits a rich, non-trivial dynamics that reflects the behaviors seen in real collective systems.

One of the most captivating aspects of the Vicsek model is the phase transition it undergoes as the key parameter of noise intensity crosses a threshold. At low noise levels, the system tends to exhibit ordered behavior, with all particles moving in a coordinated manner; conversely, at high noise levels, the directions of particles become randomized, resulting in a disordered state. Understanding the precise mechanisms behind this phase transition remains a crucial challenge in the fields of statistical physics and complex systems \cite{Vicsek1995,Sumpter2006}.
%

Since the foundational work laid by Vicsek, subsequent research has often focused on empirical or numerical studies of the model \cite{Vicsek2012}. However, the theoretical analysis of the Vicsek model poses significant challenges, which are seldom addressed comprehensively in the literature. Traditionally, mathematical studies focusing on synchronization within the Vicsek model have tended to overlook the impact of noise \cite{Jadba2003, Tang2007, Liu2009, Chen2014,Zheng2017}. A pivotal advancement in the analysis of the original model was presented by Chen in \cite{Chen2017}, where it was shown that even minor noise can disrupt the collective behavior of the Vicsek model. This theoretical insight is striking, as it seemingly contradicts empirical observations suggesting that noise of small amplitude can actually enhance synchronization within the model. This paradox underscores the need for a more thorough and deeper analysis.

In this paper, we delve deeper into the characteristics of the model. We find that the destructive effects of noise identified in \cite{Chen2017} are only of a probability sense, i.e., \emph{almost surely}. By interpreting another probability sense-namely synchronization \emph{in mean}-we rigorously establish, for any initial configuration, there exists a threshold of noise amplitude that guarantees synchronization in the Vicsek model. To be specific, we demonstrate that regardless of the number of agents, their initial positions, movement directions, speeds, or neighborhood radius, there exists a bound on noise amplitude such that if the noise amplitude is kept within this bound, the system will achieve synchronization in mean. Since the Vicsek model cannot achieve synchronization without noise for some initial states, these results indicate that random noise is crucial for the emergence of order from disorder in the Vicsek model.

The remainder of this paper is organized as follows: Section \ref{Section:model} introduces the Vicsek model; Section \ref{Section:results} presents our main findings; Section \ref{Section:simulations} provides some simulation results to validate the theoretical results and Section \ref{Section:conclusions} concludes with some final remarks.


\section{Model}\label{Section:model}
Let the position space $M=[0,B]^2\subset\mathbb{R}^2$ of the model is a square with length of side $B>0$, and $\partial M=\{\partial M_1,\ldots,\partial M_4)\}$ are four sides of $M$.
For any $y\in\mathbb{R}$, denote
\begin{equation}\label{Equa:boundednotation}
  y_{[0,B]}=\left\{
              \begin{array}{ll}
                0, & \hbox{$y<0$} \\
                y, & \hbox{$0\leq y\leq B$,} \\
                B, & \hbox{$y>B$}
              \end{array}
            \right.
\end{equation}
and for $y=(y^1,\ldots,y^d)\in \mathbb{R}^d(d>1)$, define $y_{[0,B]^d}=(y^1_{[0,B]},\ldots,y^d_{[0,B]})$ which means each $y^i$ satisfying (\ref{Equa:boundednotation}).

Let $\mathcal{V}=\{1,\ldots,n\}$ be the set of agents.
The Vicsek model takes the following rules:
\begin{equation}\label{Model:Vicsekmodel}
\begin{cases}
  \theta_i(t+1)=&\frac{1}{|\mathcal{N}_i(t)|}\sum_{j\in\mathcal{N}_i(t)}\theta_j(t)+\xi_i(t+1)\\
  x_i(t+1)=&(x_i(t)+v(\cos\theta_i(t+1),\sin\theta_i(t+1))^T)_{[0,B]^2}
\end{cases},
\end{equation}
where $\theta_i(t)\in[-\pi,\pi)$ is the heading angles of agent $i$ at time $t$ which determines the direction of its movement, $x_i(t)\in[0,B]^2$ denotes the position of agent $i$ at time $t$ within the bounded region $M$, $v>0$ is the constant moving speed and $\xi_i(t)$ represents the random noise affecting the heading angle, typically modeled as a random variable with a specified distribution. $|\cdot|$ is the cardinality of a set. In addition,
\begin{equation}\label{Model:neighbor}
  \mathcal{N}_i(t)=\{j\in\mathcal{V}:\|x_i(t)-x_j(t)\|\leq r\}
\end{equation}
is the set of neighbors of agent $i$, and $r\in(0,B)$ is the neighbor radius of agent $i$. Here, $\|\cdot\|$ is the Euclidean norm.

In order to study the noise-driven synchronization of the model (\ref{Model:Vicsekmodel}) - (\ref{Model:neighbor}), we introduce the definition of $\tau$-synchronization \emph{almost surely} and \emph{in mean}.
\begin{defn}\label{Def:meansyn}
Denote
$  d_{\theta}(t)=\max\limits_{i, j\in \mathcal{V}}|\theta_i(t)-\theta_j(t)|, t\geq 0$. For any $\tau>0,$
\begin{enumerate}
  \item if $\mathbb{P}\Big\{\limsup\limits_{t\rightarrow\infty}d_\theta(t)\leq \tau\Big\}=1$, we say the system (\ref{Model:Vicsekmodel}) - (\ref{Model:neighbor}) achieves $\tau$-synchronization almost surely (a.s.);
  \item if $\limsup\limits_{t\rightarrow\infty}\mathbb{E}\,d_\theta(t)\leq \tau$, we say the system (\ref{Model:Vicsekmodel}) - (\ref{Model:neighbor}) achieves $\tau$-synchronization in mean (i.m.).\\
\end{enumerate}
\end{defn}

\section{Main Results}\label{Section:results}

From the theory of \cite{Chen2017}, the Vicsek model (\ref{Model:Vicsekmodel}) - (\ref{Model:neighbor}) cannot achieve $\tau$-synchronization a.s.
\begin{thm}\label{Thm:noasyncrhnoise}
Given any configuration $x_i(0)\in [0,B]^2$, $\theta_i(0)\in[-\pi,\pi)$, $i\in\mathcal{V}$ and $v>0, r\in(0,B)$ of the systems (\ref{Model:Vicsekmodel}) - (\ref{Model:neighbor}), suppose the noises $\{\xi_i(t)\}_{i\in\mathcal{V},t\geq 1}$ are zero-mean random variables with independent and identical distribution (i.i.d.), and $\mathbb{E}\,\xi^2_1(1)>0, |\xi_1(1)|\leq \delta$ a.s. for $\delta>0$.
Then for any $0<\tau<2\pi$ and $\delta>0$, the Vicsek model (\ref{Model:Vicsekmodel}) - (\ref{Model:neighbor}) cannot achieve $\tau$-synchronization a.s., and actually we have
\begin{equation}\label{Equa:Pdthetasyn}
  \mathbb{P}\Big\{\limsup\limits_{t\rightarrow\infty}d_\theta(t)> \tau\Big\}=1.
\end{equation}
\end{thm}
Theorem \ref{Thm:noasyncrhnoise} shows that the Vicsek model cannot achieve synchronization for any initial conditions and noise amplitude, which seems contradict people's practical observations via simulation or experiment studies of Vicsek model. Actually, this paradox arises only because of the mathematical definition of synchronicity in the theorem. If we adjust the mathematical definition of synchronicity, we can have the following result:
\begin{thm}\label{Thm:meansyncrhnoise}
Given any configuration $x_i(0)\in [0,B]^2$, $\theta_i(0)\in[-\pi,\pi)$, $i\in\mathcal{V}$ and $v>0, r\in(0,B)$ of the systems (\ref{Model:Vicsekmodel}) - (\ref{Model:neighbor}), suppose the noises $\{\xi_i(t)\}_{i\in\mathcal{V},t\geq 1}$ are i.i.d. zero-mean random variables with $\mathbb{E}\,\xi^2_1(1)>0, |\xi_1(1)|\leq \delta$ a.s. for $\delta>0$.
Then for any $\tau>0$, there exists $\bar{\delta}=\bar{\delta}(n,r,v,B,\tau)>0$ such that for any $\delta\in(0,\bar{\delta}]$, the Vicsek model (\ref{Model:Vicsekmodel}) - (\ref{Model:neighbor}) can achieve $\tau$-synchronization i.m., i.e.,
\begin{equation}\label{Equa:Edthetasyn}
  \limsup\limits_{t\rightarrow\infty}\mathbb{E}\,d_\theta(t)<\tau.
\end{equation}
\end{thm}
\begin{rem}\label{Remark:DiffThm12}
The difference between these two theorems is quite mathematical. Practically, when the Vicsek model achieves synchronization at some moment $T>0$, Theorem \ref{Thm:noasyncrhnoise} means that, one can always observe the synchronized Vicsek model divided at some moment $t>T$ in an experiment; while Theorem \ref{Thm:meansyncrhnoise} reveals that if the experiment is repeated many times, the average angle difference $d_\theta(t)$ at any time $t>T$ is still tiny.
\end{rem}

The proof of Theorem \ref{Thm:meansyncrhnoise} is based on the following lemmas.
\begin{lem}\label{Lem:enterbarD}
Given any configurations of the systems (\ref{Model:Vicsekmodel}) - (\ref{Model:neighbor}) in Theorem \ref{Thm:meansyncrhnoise} and $0<\epsilon<r$, let $T_\epsilon=\inf\{t:\max_{i,j\in\mathcal{V}}\|x_i(t)-x_j(t)\|<\epsilon\}$. Then there exists $\bar{\delta}>0$ such that $\mathbb{P}\{T_\epsilon<\infty\}=1$ for all $\delta\in(0,\bar{\delta}]$.
\end{lem}
\begin{proof}
Since $\{\xi_i(t), i\in\mathcal{V}, t>0\}$ are nondegenerate random variables of zero-mean, it is easy to know that for any $\delta\in(0,2\pi]$ there exist $0<p<1$, $0<\alpha<\delta$ such that
\begin{equation}\label{Equa:probxiap}
  \mathbb{P}\{\xi_i(t)>\alpha\}>p,~~~\mathbb{P}\{\xi_i(t)<-\alpha\}>p.
\end{equation}

Based on the theory of \cite{Chen2017}, it is known that if we can find a ``control series'' based on noise which drives the system to a targeted position with a uniformly positive probability in a uniformly finite time, the system can achieve the targeted position in finite time almost surely. Now for any $x(0)\in [0,B]^{2\times n}$, $\theta(t)\in[-\pi,\pi)^n$, we can design the ``control series'' of noise as follows.

At any time $t\geq 0$, the agents can be divided into $w(1\leq w\leq n)$ connected components. Take anyone connected component at $t$, say $C_1$ for simplicity, rename the agents according to the order of angle $\theta_1(t)\leq \theta_2(t)\leq\ldots\leq\theta_w(t)$($w$ is number of the agents in $C_1$) and denote $\theta_{mid}(t)=\frac{\theta_1(t)+\theta_w(t)}{2}$, then we can design the ``control series'' as follows:
\begin{equation}\label{Equa:noisecontrolseriesangle}
\xi_i(t)\in  \begin{cases}
    & (\alpha,\delta), \qquad\,\,\,\text{if}\quad\theta_i(t)<\theta_{mid}(t),  \\
      & (-\delta,-\alpha),\quad \text{if}\quad\theta_i(t)\geq\theta_{mid}(t).
  \end{cases}
\end{equation}
Under (\ref{Equa:noisecontrolseriesangle}), by (\ref{Model:Vicsekmodel}) and (\ref{Equa:probxiap}) we can know that, for some sufficiently small $\bar{\delta}_1 >0$ and any given $\delta\in(0,\bar{\delta}_1]$, there exists $0<\alpha\leq \delta$ so that $|\theta_w(1)-\theta_1(1)|$ decreases at least by $2\alpha>0$ after each step, then there exists $0<L_0<\max\{\frac{2\pi-d_\theta(0)}{2\alpha},\frac{2B}{v}\}$ such that $|\theta_w(L_0)-\theta_1(L_0)|\leq 2\delta$ and $x_i(L_0)\in \partial M, i=1,\ldots,w$, i,e., all agents of the components are on the boundary. For $t\geq L_0$, rename the agents according to the order of their position on $\partial M$ as $x_1(t)\rightarrow\ldots\rightarrow x_w(t)$ in the anticlockwise direction ``$\rightarrow$'' along the boundary, and the ``control series'' based on noise can be taken as
\begin{equation}\label{Equa:noisecontrolseriesangleposi}
 \xi_i(t)\in \begin{cases}
    & (\alpha,\delta), \qquad\,\,\,\text{if}\quad i<\frac{w}{2},  \\
      &(-\delta,-\alpha),\quad \text{if}\quad i\geq\frac{w}{2}.
  \end{cases}
\end{equation}
Under (\ref{Equa:noisecontrolseriesangleposi}), by (\ref{Model:Vicsekmodel}) and (\ref{Equa:probxiap}) we can know that, for some sufficiently small $0<\bar{\delta}_2$ and any given $\delta\in(0,\bar{\delta}_2]$, there exists $0<\alpha\leq\delta$ and $0<L_1\leq \frac{\pi}{4\alpha}$ such that the angle between the move direction of agents and the boundary is less than $\frac{\pi}{4}$ at $L_0+L_1$, then again by (\ref{Model:Vicsekmodel}),
there exists $0<L_2\leq \frac{2B}{v}$ such that for all $\delta\in(0,\frac{\pi}{4L_2}]$,
\begin{equation}\label{Equa:posiepsiL0L1}
  \max_{i,j\in C_1}\|x_i(L_0+L_1+L_2)-x_j(L_0+L_1+L_2)\|<\epsilon.
\end{equation}
Carry on the above design for each connected components, then there exists $L_w>0$ so that at $t=L_w$, the group of agents can be divided into $w(1\leq w\leq n)$ subgroups, and  $\max_{i,j\in\mathcal{V}_k}\|x_i(L_w)-x_j(L_w)\|<\epsilon, k=1,\ldots,w$.

Now we will prove for this special state configuration, we can design a ``control series'' of noise. First, suppose $w=2$ and without loss of generality $L_w=0$ for simplicity, then we can design a ``control series'' of noise as follows. Since $\max_{i,j\in\mathcal{V}_k}\|x_i(0)-x_j(0)\|<\epsilon, k=1,2$, by (\ref{Model:Vicsekmodel}), $\max_{i,j\in\mathcal{V}_k}|\theta_i(1)-\theta_j(1)|\leq 2\delta$. Let $\tilde{\theta}_{mid}(t)=\frac{\tilde{\theta}_1(t)+\tilde{\theta}_2(t)}{2}$ where $\tilde{\theta}_i(t)=\frac{1}{|\mathcal{V}_i(t)|}\sum_{k\in\mathcal{V}_i}\theta_k(t)$, $i=1,2$,
\begin{equation}\label{Equa:contrtwogroup}
\xi_i(t)\in
\begin{cases}
  &(\alpha,\delta), \qquad\,\,\,\text{if}\quad\theta_i(t)<\tilde{\theta}_{mid}(t),  \\
      & (-\delta,-\alpha),\quad \text{if}\quad\theta_i(t)\geq\tilde{\theta}_{mid}(t).
\end{cases}
\end{equation}
Under (\ref{Equa:contrtwogroup}), by selecting sufficiently small $\epsilon$, $\delta$ and using a similar argument of (\ref{Equa:posiepsiL0L1}), we can know that there exists $L>0$ such that $\max_{i,j\in\mathcal{V}}\|x_i(L)-x_j(L)\|<\epsilon$. This prove the case of $w=2$.
Suppose the conclusion holds for $w=k<n$. Now we consider the case of $w=k+1$. The conclusion is immediate since we can select any two subgroups and use (\ref{Equa:noisecontrolseriesangle}) and (\ref{Equa:contrtwogroup}) repeatedly untitle the number of subgroups reduce to $w\leq k$ within a finite period. In conclusion, we can finally find a ``control series'' based on noise within a finite definite time $L>0$ such that $\max_{i,j\in\mathcal{V}}\|x_i(L)-x_j(L)\|<\epsilon$.
This completes the proof.
\end{proof}
Lemma \ref{Lem:enterbarD} shows that for any initial configuration, the system can always achieve a state where all agents are arbitrarily close to each other in a finite time.
In the following, denote $d_x(t)=\max_{i,j\in\mathcal{V}}\|x_i(t)-x_j(t)\|,\,t\geq 0$.

\begin{lem}\label{Lem:probdXtrho}
Given any $\theta_0\in[-\pi,\pi)^n$, $0<a\leq r$, if $d_x(0)\leq \frac{a}{3}$, then for any $\rho>0$, there exists $0<\bar{\delta}\leq \frac{\tau}{2}$ such that $\mathbb{P}\{d_x(t)>a\}<\rho$ for all $t>0$ whenever $\delta\in(0,\bar{\delta}]$.
\end{lem}
\begin{proof}
First we prove for any $x(0)\in [0,B]^{n\times2}$ with $d_x(0)\leq \frac{2a}{3}$, there exists $L>0, 0<p<1$ and $\bar{\delta}_0>0$ such that $d_x(t)\leq a$, $0<t\leq L$ a.s. and $\mathbb{P}\{d_x(L)\leq \frac{a}{3}\}>p$ for $0<\delta\leq \bar{\delta}_0$.

Denote $T_a=\inf\{t:d_x(t)>a\}$.  Since $d_x(0)\leq \frac{2a}{3}<\frac{2r}{3}$, by (\ref{Model:Vicsekmodel}) we have for $0\leq t<T_a$,
\begin{equation}\label{Equa:thetataylor}
\begin{split}
  \theta_i(t+1)=&\frac{1}{n}\sum_{j\in\mathcal{V}}\theta_j(t)+\xi_i(t+1)\\
  &\cdots\\
  =&\frac{1}{n}\sum_{j\in\mathcal{V}}\theta_j(0)+\sum_{k=1}^t\frac{1}{n}\sum_{i\in\mathcal{V}}\xi_i(k)+\xi_i(t+1)\\
  =&\tilde\theta(0)+\sum_{k=1}^t\frac{1}{n}\sum_{i\in\mathcal{V}}\xi_i(k)+\xi_i(t+1)
  \end{split}
\end{equation}
where $\tilde\theta(t)=\frac{1}{|\mathcal{N}_i(t)|}\sum_{j\in\mathcal{N}_i(t)}\theta_j(t)=\frac{1}{n}\sum_{j\in\mathcal{V}}\theta_j(t)$,
and then
\begin{equation}\label{Equa:xittaylor}
\begin{split}
  &x_i(t+1)=x_i(t)+v(\cos\theta_i(t+1),\sin\theta_i(t+1))^T\\
  =&x_i(t)+v(\cos\tilde\theta(0),\sin\tilde\theta(0))^T+v(-\sin\tilde\theta(0),\cos\tilde\theta(0))^T\Big(\sum_{k=1}^t\bar{\xi}(k)+\xi_i(t+1)+o(\delta)\Big)\\
  &\cdots\\
  =&x_i(0)+v(t+1)(\cos\tilde\theta(0),\sin\tilde\theta(0))^T\\
  &+v(-\sin\tilde\theta(0),\cos\tilde\theta(0))^T\sum_{s=0}^t\Big(\sum_{k=1}^s\bar{\xi}(k)+\xi_i(s+1)+o(\delta)\Big)
\end{split}
\end{equation}
where $\bar{\xi}(k)=\frac{\sum_{i\in\mathcal{V}}\xi_i(k)}{n}$,  then
\begin{equation}\label{Equa:disx1x2}
\begin{split}
&d_x(t+1)=\max_{i,j\in\mathcal{V}}\|x_i(t+1)-x_j(t+1)\|\\
=&\max_{i,j\in\mathcal{V}}\Big\|x_i(t)-x_j(t)+v(-\sin\tilde\theta(0),\cos\tilde\theta(0))^T(\xi_i(t+1)-\xi_j(t+1)+o(\delta))\Big\|\\
  =&\max_{i,j\in\mathcal{V}}\Big\|x_i(0)-x_j(0)+v(-\sin\tilde\theta(0),\cos\tilde\theta(0))^T\sum_{k=0}^{t}(\xi_i(k+1)-\xi_j(k+1)+o(\delta))\Big\|.
\end{split}
\end{equation}

Denote $T_b=\inf\{t: agents~ reach~ the~ boundary\}$.
By (\ref{Equa:xittaylor}), there exists some $\bar{\delta}_1>0$, when $0<\delta<\bar{\delta}_1$, agent $i$ will move towards the boundary by at least $\frac{v}{2}$ after each step, hence
\begin{equation}\label{Equa:Treachbound}
  T_b\leq \frac{4B}{v},~ a.s.
\end{equation}
Meanwhile by (\ref{Equa:disx1x2}),
\begin{equation}\label{Equa:dxtijinequity}
  \begin{split}
d_x(t+1)=&\max_{i,j\in\mathcal{V}}\Big\|x_i(t)-x_j(t)+v(-\sin\tilde\theta(0),\cos\tilde\theta(0))^T(\xi_i(t+1)-\xi_j(t+1)+o(\delta))\Big\|\\
  \leq&d_x(0)+\max_{i,j\in\mathcal{V}}v\|(-\sin\tilde\theta(0),\cos\tilde\theta(0))^T\sum_{k=0}^{t}(\xi_i(t+1)-\xi_j(t+1)+o(\delta))\|.
\end{split}
\end{equation}
Since $d_x(0)\leq \frac{2a}{3}$, by (\ref{Equa:dxtijinequity}) there exists $0<\bar{\delta}_2<\frac{a}{24B}$ such that when $0<\delta<\bar{\delta}_2$,
\begin{equation}\label{dxtleqa}
  d_x(t)\leq a,~~~a.s.
\end{equation}
for $0\leq t\leq \frac{4B}{v}$.
Thus when $0<\delta\leq \min(\bar{\delta}_1, \bar{\delta}_2)$, by (\ref{Equa:Treachbound}) and (\ref{dxtleqa}) we have
\begin{equation}\label{Equa:dxTgeqr}
\begin{split}
 \mathbb{P}\{d_x(T_b)>a\}=&\sum_{k=1}^\infty\mathbb{P}\{d_x(k)>a, T_b=k\} \\
 \leq&\sum_{k=1}^{\lfloor\frac{4B}{v}\rfloor}\mathbb{P}\{d_x(k)>a, T_b=k\}+\sum_{k=\lfloor\frac{4B}{v}\rfloor+1}^\infty\mathbb{P}\{T_b=k\}\\
 = &0.
 \end{split}
\end{equation}
Hence $T_b\leq T_a$ a.s.
Denote $i_0$ as the agent who achieves the boundary at $T_b$. Let
\begin{equation*}
  \textbf{x}_{i,j}(t)=x_i(t)-x_j(t)
\end{equation*}
be the vector between the positions of agents $i$ and $j$ at $t$, and
\begin{equation*}
  \phi(t)=\max_j\langle\textbf{x}_{(i_0,j)}(t), (\cos\theta_j(t),\sin\theta_j(t))\rangle
\end{equation*}
be the angle between the position vector and move direction at $t$ (See Figure \ref{Fig:phit}).

\begin{figure}
    \centering
 \includegraphics[width=2.5in]{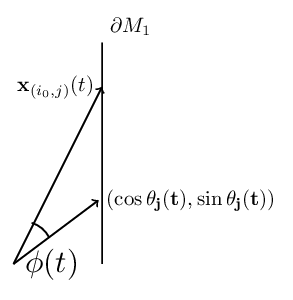}\\
\caption{$\phi(t)$} 
    \label{Fig:phit} 
\end{figure}

From (\ref{Model:Vicsekmodel}), we can see that when $i_0$ first reach the boundary, it will stop at the boundary with its original direction, while the other agents still move towards the boundary by at least $\frac{v}{2}$ with the original direction. Then by (\ref{Equa:xittaylor}), there exists $\phi_0>0$ and $\bar{\delta}_3>0$ such that when $\phi(T_b)\leq \phi_0$ and $\delta<\bar{\delta}_3$, there exists $L_0>0$ such that all agents reach the boundary at $t=T_b+L_0$ and
\begin{equation}\label{Equa:dxtTbL01}
  d_x(T_b+L_0)\leq \frac{a}{3}.
\end{equation}
While if $\phi(T_b)>\phi_0$, we can select the trajectories of $\xi(t)$ which happens with a positive probability under which conclusions like (\ref{Equa:dxtTbL01}) hold, i.e., there exists $L_1>0$, $\bar{\delta}_4>0$ and a nonempty set $A\in [-\delta,\delta]^{n\times L_1}$ such that when $\{\xi(T_b+t), 1\leq t\leq L_1\}\in A$ and $\delta<\bar{\delta}_4$, we have
\begin{equation}\label{Equa:dxtTbL02}
  d_x(T_b+L_1)\leq \frac{a}{3}.
\end{equation}
Denote $S_A(T,L_1)=\{\xi(T+t)\in A, 1\leq t\leq L_1\}$ for any stopping time $T$, by (\ref{Equa:probxiap}) and i.i.d. of $\{\xi_i(t), i\in\mathcal{V}, t\geq 1\}$ and hence strong Morkov property, we know
\begin{equation}\label{Equa:probSAL1}
  \mathbb{P}\{S_A(T,L_1)\}=\mathbb{P}\{S_A(0,L_1)\}>0.
\end{equation}
Take $L'=\frac{4B}{v}+\max(L_0, L_1)$, $\bar{\delta}_0=\min(\bar{\delta}_1,\bar{\delta}_2,\bar{\delta}_3,\bar{\delta}_4)/L$, $p\in(0,\mathbb{P}\{S_A(T_b,L_1)\})$, by (\ref{Equa:Treachbound}), (\ref{Equa:dxtTbL01}), (\ref{Equa:dxtTbL02}) and (\ref{Equa:probSAL1}), we prove the assert at the beginning, i.e., for any $x(0)\in [0,B]^{n\times 2}$ with $d_x(0)\leq \frac{2a}{3}$, there exists $\max(L_0, L_1)\leq L\leq L', 0<p<1$ and $\bar{\delta}_0>0$ such that for all $\delta<\bar{\delta}_0$,
\begin{equation}\label{Equa:dxLgeqp}
  \begin{split}
  &d_x(t)\leq a,~0<t\leq L, ~a.s.,\\
  &\mathbb{P}\Big\{d_x(L)\leq \frac{a}{3}\Big\}\geq\mathbb{P}\{S_A(T_b,L_1)\}>p.
  \end{split}
\end{equation}
The first inequality of (\ref{Equa:dxLgeqp}) shows that given $d_x(0)<\frac{2a}{3}$ and $\delta\in(0,\bar{\delta}_0)$, $d_x(t)$ can only increase no more than $\frac{a}{3}$ during a period of length $L$. Subsequently, given any $0<\rho<1$, let $K=\lceil \log_{1-p}\rho\rceil$, $\bar{\delta}=\frac{\bar{\delta}_0}{K}$, also considering (\ref{Equa:dxtijinequity}), we know that
whenever $d_x(0)\leq \frac{a}{3}$, $0<\delta\leq \bar{\delta}$, it follows
\begin{equation}\label{Equa:dxtKL}
\begin{split}
  &d_x(t)\leq \frac{2a}{3}, \,\,\quad 0<t\leq KL,~ a.s.\\
  &d_x(t)\leq a, \qquad KL<t\leq 2KL,~ a.s.
\end{split}
\end{equation}
and if during the period between $KL$ and $2KL$, $S_A(s,L_1)$ occurs once for some $KL\leq s\leq (2K-1)L$, it follows that $d_x(t)\leq \frac{2a}{3}$ for $s+L\geq t\leq 2KL$ since $d_x(s+L)\leq \frac{a}{3}$ by (\ref{Equa:dxLgeqp}). Hence if $d_x(t)> \frac{2a}{3}$ for some moment $KL<t\leq 2KL$, it must happens that $S_A(kL,L_1)$ does not occur for all $K\leq k\leq 2K-1$. Considering the  independence of $\{\xi_i(t), i\in\mathcal{V}, t\geq 1\}$, by (\ref{Equa:probSAL1}) and the last inequality of (\ref{Equa:dxLgeqp}), we have for $KL<t\leq 2KL$,
\begin{equation}\label{Equa:dxtgeqrho}
\begin{split}
  \mathbb{P}\Big\{d_x(t)>\frac{2a}{3}\Big\}\leq &\prod_{k=K}^{2K-1}(1-\mathbb{P}\{S_A(kL,L_1)\})\\
  <& (1-p)^K\leq\rho.
  \end{split}
\end{equation}
Define stopping times: $T_0^\alpha=0$, $T_{k}^\beta=\inf\{t\geq T_k^\alpha:d_x(t)\geq\frac{2a}{3}-2v\delta\}$, $T_{k+1}^\alpha=\inf\{t\geq T_k^\beta:d_x(t)\leq\frac{a}{3}\}$, $k\geq 0$. Here, $\delta>0$ can certainly be selected to be small enough such that $\frac{2a}{3}-2v\delta>\frac{a}{3}$. By Theorem \ref{Thm:noasyncrhnoise} and Lemma \ref{Lem:enterbarD}, we have
\begin{equation*}
  \mathbb{P}\{T_k^\alpha<\infty\}=\mathbb{P}\{T_k^\beta<\infty\}=1, \quad k\geq 0.
\end{equation*}
By (\ref{Equa:dxtKL}), (\ref{Equa:dxtgeqrho}) and strong Markov property, for all $k\geq 0$,
\begin{equation}\label{Equa:stoprho}
\begin{split}
  &\mathbb{P}\{T_{k+1}^\alpha-T_k^\beta>KL\}<\rho,\\
  &\mathbb{P}\{d_x(t)>a|T_{k+1}^\alpha-T_k^\beta\leq KL, T_k^\beta\leq t<T_{k+1}^\alpha\}=0.
 \end{split}
\end{equation}
By (\ref{Equa:stoprho}),
\begin{equation}\label{Equa:dxtrhostopdif1}
\begin{split}
  &\mathbb{P}\{d_x(t)>a,T_{k+1}^\alpha-T_k^\beta>KL, T_k^\beta\leq t<T_{k+1}^\alpha\}\\
  =&\mathbb{P}\{d_x(t)>a, T_k^\beta\leq t<T_{k+1}^\alpha|T_{k+1}^\alpha-T_k^\beta>KL\}\mathbb{P}\{T_{k+1}^\alpha-T_k^\beta>KL\}\\
  <&\rho\mathbb{P}\{d_x(t)>a, T_k^\beta\leq t<T_{k+1}^\alpha|T_{k+1}^\alpha-T_k^\beta>KL\}\\
  =&\rho\mathbb{P}\{d_x(t)>a| T_k^\beta\leq t<T_{k+1}^\alpha,T_{k+1}^\alpha-T_k^\beta>KL\}\mathbb{P}\{T_k^\beta\leq t<T_{k+1}^\alpha\}\\
  \leq&\rho\mathbb{P}\{T_k^\beta\leq t<T_{k+1}^\alpha\},
\end{split}
\end{equation}
and also
\begin{equation}\label{Equa:dxtrhostopdif2}
  \mathbb{P}\{d_x(t)>a,T_{k+1}^\alpha-T_k^\beta\leq KL, T_k^\beta\leq t<T_{k+1}^\alpha\}=0.
\end{equation}
Hence by (\ref{Equa:dxtrhostopdif1}) and (\ref{Equa:dxtrhostopdif2}),
\begin{equation}\label{Equa:dxtrho1}
\begin{split}
  &\mathbb{P}\{d_x(t)>a, T_k^\beta\leq t<T_{k+1}^\alpha\}\\
  =&\mathbb{P}\{d_x(t)>a,T_{k+1}^\alpha-T_k^\beta\leq KL, T_k^\beta\leq t<T_{k+1}^\alpha\}+\mathbb{P}\{d_x(t)>a,T_{k+1}^\alpha-T_k^\beta> KL, T_k^\beta\leq t<T_{k+1}^\alpha\}\\
  <&\rho\mathbb{P}\{T_k^\beta\leq t<T_{k+1}^\alpha\}.
\end{split}
\end{equation}
Notice for all $k\geq 0, t\geq 0$,
\begin{equation}\label{Equa:dxtrho2}
\mathbb{P}\{d_x(t)>a, T_k^\alpha\leq t<T_{k}^\beta\}=0.
\end{equation}
Therefore, given any $t>0$, by (\ref{Equa:dxtrho1}) and (\ref{Equa:dxtrho2}), we obtain
\begin{equation}\label{Equa:dvtrho}
\begin{split}
  \mathbb{P}\{d_x(t)>a\}=&\sum_{k\geq 0}\mathbb{P}\{d_x(t)>a, T_k^\alpha\leq t<T_{k}^\beta\}+\mathbb{P}\{d_x(t)>a, T_k^\beta\leq t<T_{k+1}^\alpha\} \\
    <&\rho\sum_{k\geq 0}\mathbb{P}\{T_k^\beta\leq t<T_{k+1}^\alpha\}\\
    \leq&\rho.
\end{split}
\end{equation}
This completes the proof.
\end{proof}
Lemma \ref{Lem:probdXtrho} shows that once the system reaches a state where all agents are neighbors close enough to each other, the probability for the agents to get away far from each other can be arbitrarily small as long as noise amplitude is sufficiently small.

\begin{lem}\label{Lem:probdvt}
For any $n\geq 2, 0<\bar{\theta}\leq \tau, \rho>0$, let $T$ be a stopping time with $d_x(T)\leq \frac{r}{3}$, then for any $t>T$ there exists $0<\bar{\delta}< \frac{\bar{\theta}}{2}$ and an $\sigma(\xi(T+1),\ldots,\xi(T+t))$-measurable event $S_t$ with $ \mathbb{P}\{S_t\}<\rho$ such that $\{d_\theta(T+t+1)>\bar{\theta}\}\subset S_t$ whenever $\delta\in(0,\bar{\delta}]$.
\end{lem}
\begin{proof}
For the simplicity of notations in the following proof, suppose $T=0$ a.s. without loss of generality due to the fact that $\{\xi_i(t), i\in\mathcal{V}, t\geq 1\}$ are i.i.d. and hence strong Markov property holds here.

Given $\rho>0$, let $\bar{\delta}_1$ be the $\bar{\delta}$ in Lemma \ref{Lem:probdXtrho}, and $a=r, L_\rho=KL$ in the proof there. By (\ref{Model:Vicsekmodel}) and (\ref{Equa:thetataylor}),
we know whenever $d_x(t)\leq r$,
\begin{equation}\label{Equa:dthetal2delta}
  d_\theta(t+1)\leq 2\delta<\bar{\theta}, ~~~a.s.
\end{equation}
Take $\bar{\delta}=\min\{\bar{\delta}_1, \frac{\bar{\theta}}{2L_\rho}\}$. Since $d_x(0)\leq \frac{r}{3}$, by (\ref{Equa:dxtKL}) we have
\begin{equation}\label{Equa:dxtleqrLrho}
  d_x(t)\leq \frac{2r}{3}<r,\,\,\,t\leq L_\rho.
\end{equation}
For $t> L_\rho$, by Lemma \ref{Lem:probdXtrho},
\begin{equation}\label{Equa:dthetatgeq}
\begin{split}
  &\mathbb{P}\{d_x(t)>r\}<\rho,~~~t>L_\rho.
  \end{split}
\end{equation}
Let $S_t=\{d_x(t)>r\}$, by (\ref{Equa:dthetal2delta}) we have
$\{d_\theta(t+1)>\bar{\theta}\}\subset S_t$. And by the derivation process of (\ref{Equa:dxtgeqrho}), $S_t$ is $\sigma(\xi(1),\ldots,\xi(t))$-measurable. Further by  (\ref{Equa:dxtleqrLrho}) and (\ref{Equa:dthetatgeq}), $\mathbb{P}\{S_t\}<\rho$ for all $t>0$. This completes the proof.
\end{proof}

\emph{Proof of Theorem \ref{Thm:meansyncrhnoise}:}
Let $T=\inf\{t: d_x(t)\leq \frac{r}{3}\}$,
then by Lemma \ref{Lem:enterbarD}
\begin{equation}\label{Equa:Tlessinftythm}
  \mathbb{P}\{T<\infty\}=1,
\end{equation}
and
\begin{equation}\label{Equa:Edthetat1}
  \begin{split}
\mathbb{E}\,d_\theta(t+1)=&\mathbb{E}\Big(d_\theta(t+1)I_{\{T\leq t\}}+d_\theta(t+1)I_{\{T>t\}}\Big)\\
=&\mathbb{E}\bigg(\sum_{k=0}^{t}d_\theta(T+k+1)I_{\{T=t-k\}}\bigg)+\mathbb{E}\,d_\theta(t+1)I_{\{T>t\}}\\
=&\sum_{k=0}^{t}\mathbb{E}\,d_\theta(T+k+1)I_{\{T=t-k\}}+\mathbb{E}\,d_\theta(t+1)I_{\{T>t\}}.
\end{split}
\end{equation}
Take $\bar{\theta}=\frac{\tau}{2}, \rho=\frac{\tau}{4\pi}$ in Lemma \ref{Lem:probdvt}, then $\{d_\theta(T+k+1)>\frac{\tau}{2}\}\subset S_k$ for some $\sigma(\xi(T+1),\ldots)$-measurable $S_k$ and $\mathbb{P}\,\{S_k\}<\rho$. Since $d_\theta(t)\leq 2\pi$, a.s. for all $t\geq 0$, by the independence of $\sigma(T)$ and $\sigma(\xi(T+1),\ldots)$, it follows that
%
given any $t\geq 0$, we have
\begin{equation}\label{Equa:Expdthetat2}
\begin{split}
&\sum_{k=0}^{t}\mathbb{E}\,d_\theta(T+k+1)I_{\{T=t-k\}}\\
=&\sum_{k=0}^{t}\mathbb{E}\,\bigg(d_\theta(T+k+1)I_{\{d_\theta(T+k+1)\leq\frac{\tau}{2},T=t-k\}}+d_\theta(T+k+1)I_{\{d_\theta(T+k+1)>\frac{\tau}{2},T=t-k\}}\bigg)\\
\leq& \frac{\tau}{2} \sum_{k=0}^{t}\mathbb{P}\,\{T=t-k\}+2\pi\sum_{k=0}^{t}\mathbb{P}\,\Big\{d_\theta(T+k+1)>\frac{\tau}{2},T=t-k\Big\}\\
\leq& \frac{\tau}{2} \sum_{k=0}^{t}\mathbb{P}\,\{T=t-k\}+2\pi\sum_{k=0}^{t}\mathbb{P}\,\{S_k,T=t-k\}\\
=&\frac{\tau}{2}\mathbb{P}\{T\leq t\}+ 2\pi\sum_{k=0}^{t}\mathbb{P}\,\{S_k\}\mathbb{P}\{T=t-k\}\\
<&\frac{\tau}{2}\mathbb{P}\{T\leq t\}+ 2\pi\rho\mathbb{P}\{T\leq t\}\\
=&\tau\mathbb{P}\{T\leq t\},
\end{split}
\end{equation}
and
\begin{equation}\label{Equa:Edthetat3}
  \mathbb{E}\,d_\theta(t+1)I_{\{T>t\}}\leq 2\pi \mathbb{P}\{T>t\}.
\end{equation}
Then by (\ref{Equa:Tlessinftythm})-(\ref{Equa:Edthetat3}), we obtain
\begin{equation}\label{limleqepsi}
\begin{split}
  \limsup\limits_{t\rightarrow\infty}\mathbb{E}\,d_\theta(t)<& \limsup\limits_{t\rightarrow\infty}\Big(\tau \mathbb{P}\{T\leq t\}+2\pi\mathbb{P}\{T>t\}\Big)\\
  = &\tau.
\end{split}
\end{equation}
This ends the proof. $\hfill \Box$

\section{Simulations}\label{Section:simulations}
In this section, we present simulation results regarding the average of $d_\theta(t)$ to validate the main theoretical findings of this study. The model is initialized with the following parameters: $n=5, B=40, r=8$, and $v=2$. The initial angles of the agents are uniformly distributed over the interval $[-\pi/40,\pi/40]$ and $[-\pi/80,\pi/80]$, while all their initial positions are uniformly generated within the range $[0,40]^2$. We conducted 50 independent simulation runs for each noise amplitude $\delta$, and the resulting average maximum angle difference is illustrated in Figure \ref{Fig:dthetamean}. From (\ref{Equa:dthetal2delta}), we can observe that the maximum angle difference is bounded by the noise amplitude when the system reaches synchronization. Specifically, a smaller noise amplitude leads to a smaller angle difference. This relationship is also clearly demonstrated in Figure \ref{Fig:dthetamean}.

%
\begin{figure}
  \centering
  \includegraphics[width=3in]{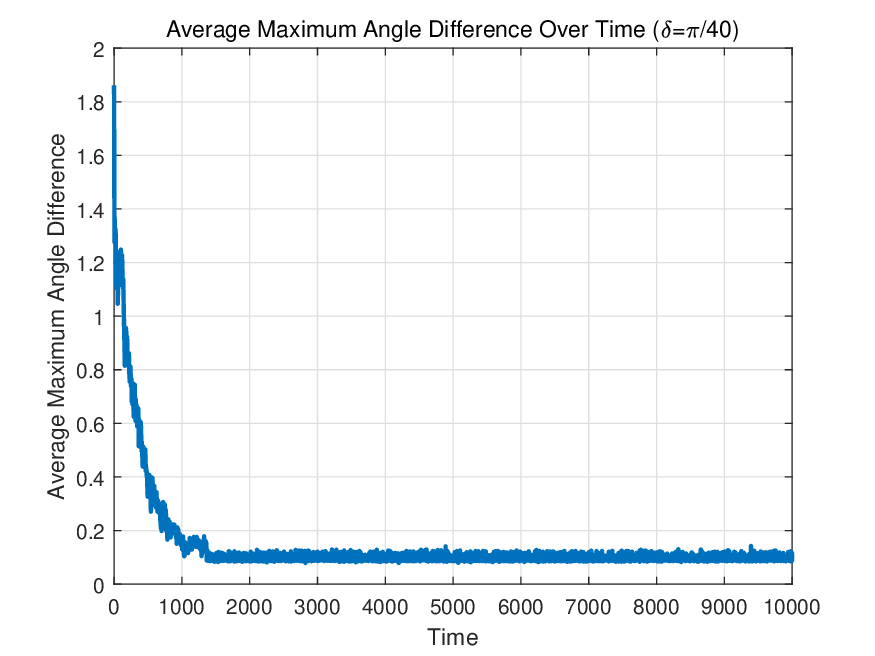}
  \includegraphics[width=3in]{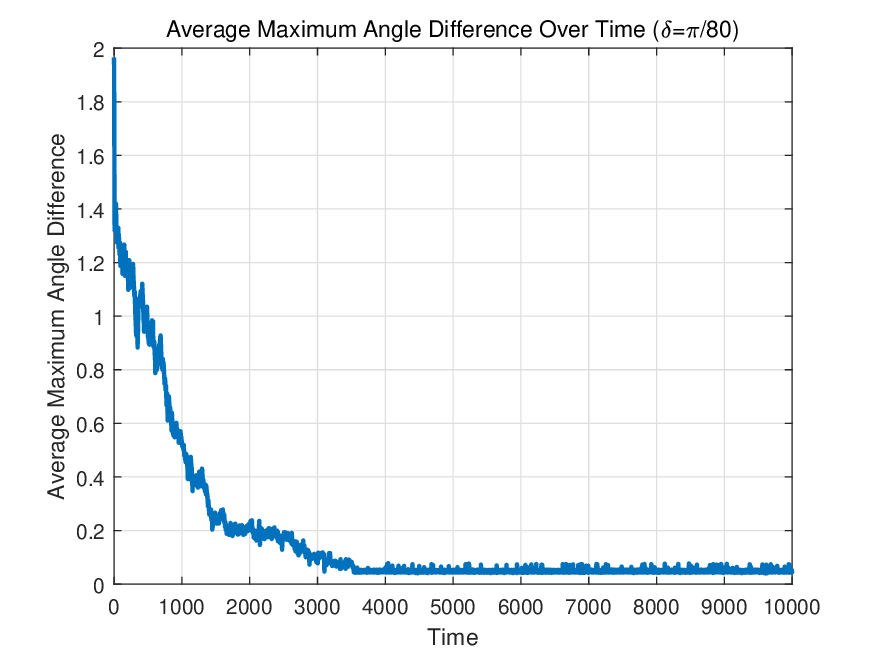}\\
  \caption{The average value of $d_\theta(t)$ across 50 simulation runs}\label{Fig:dthetamean}
\end{figure}
%
%
%

\section{Conclusions}\label{Section:conclusions}
In this study, we conducted a rigorous global analysis of the Vicsek model. It has been proven that the Vicsek model can achieve synchronization in mean when the noise amplitude is sufficiently small. Beyond the factor of noise amplitude, the boundedness of the state space is also crucial for the system to attain synchronization. In fact, through the proofs presented in this study, we can uncover the intrinsic mechanisms that lead to synchronization in the model, as well as identify additional types of synchronization. For instance, in scenarios where the noise amplitude is relatively large, if the number of agents in the group is high and their movement speed is low, the group can still coalesce and move forward cohesively. Lastly, the analytical methods and concepts provided in this research hold significant implications for the analysis of synchronization in other self-organizing systems.


\begin{thebibliography}{99}

\bibitem{Haken1983}
Haken, H.  \emph{Synergetics: An Introduction: Nonequilibrium Phase Transitions and Self-Organization in Physics, Chemistry, and Biology}, Springer, 1983.

\bibitem{Toner1995}
J.  Toner and Y. Tu, ``Long-range Order in a Two-Dimensional Dynamical XY Model: How Birds Fly Together'', \emph{Physical Review Letters}, 75(23), 4326-4329, 1995.

\bibitem{Couzin2005}
I. D. Couzin, J. Krause, D. W. Franks and S. A. Levin, ``Effective leadership and decision-making in animal groups on the move'', \emph{Nature}, 433, 513-516, 2005.

\bibitem{Bialek2012}
W. Bialek, A. Cavagna, I. Giardina, T. Mora, E. Silvestri, M. Viale and A. Walczak, ``Statistical mechanics for natural flocks of birds'', \emph{Proc. Nat. Aca. Sci.}, 109(13), 4786-4791, 2012.

\bibitem{Lowen2023}
L. Caprini and H. L\"{o}wen, ``Flocking without alignment interactions in attractive active Brownian particles'', \emph{Phys. Rev. Lett.} 130, 148202, 2023.

\bibitem{Yan2024}
T. Gao, B. Barzel and G. Yan, ``Learning interpretable dynamics of stochastic complex systems from experimental data'', \emph{Nat. Commun.}, 15, 6029, 2024.

\bibitem{Vicsek1995}
T. Vicsek, A. Czirok, E. Ben-Jacob, I. Cohen and O. Shochet, ``Novel type of phase transition in a system of self-driven particles'', \emph{Phys. Rev. Lett.}, 75(6): 1226-1229, 1995.

\bibitem{Sumpter2006}
D. Sumpter, ``The principles of collective animal behavior'', \emph{Phi. Trans. Roy. Soci. B: Bio. Sci.}, 361, 5-22, 2006.

\bibitem{Vicsek2012}
T. Vicsek and A. Zafeiris, ``Collective motion'', \emph{Phys. Rep.}, 517(3), 71-140, 2012.

\bibitem{Jadba2003}
A. Jadbabaie, J. Lin, and A. Morse, ``Coordination of groups of mobile autonomous agents using nearest neighbor rules'', \emph{IEEE Trans. Autom. Control}, 48(9), 988-1001, Sep. 2003.

\bibitem{Tang2007}
G. G. Tang and L. Guo, ``Convergence of a class of multi-agent systems in probabilistic framework'', \emph{J. Syst. Sci. Complex}, 20(2), 173-197, 2007.

\bibitem{Liu2009}
Z.X.Liu and L.Guo, ``Synchronization of multi-agent systems without connectivity assumption'', \emph{Automatica}, 45(12), 2744-2753, 2009.

\bibitem{Chen2014}
G. Chen, Z. X. Liu and L. Guo, ``The smallest possible interaction radius for synchronization of self-propelled particles'', \emph{SIAM Rev.}, 56(3), 499-521, 2014.

\bibitem{Zheng2017}
J. Zheng, J. G. Dong and L. Xie, ``Synchronization of the delayed Vicsek model'', \emph{IEEE Trans. Auto. Contr.}, 62(11), 5866--5872, 2017.

\bibitem{Chen2017}
G. Chen, ``Small noise may diversify collective motion in Vicsek model'', \emph{IEEE Trans. Auto. Contr.}, 62(2), 636-651, 2017.


%
%
%
\end{thebibliography}
\end{document}